\documentclass[sigconf,anonymous=false,balance=false]{acmart}
\copyrightyear{2025}
\acmYear{2025}
\setcopyright{cc}
\setcctype{by-nd}
\acmConference[SIGIR '25]{Proceedings of the 48th International ACM SIGIR Conference on Research and Development in Information Retrieval}{July 13--18, 2025}{Padua, Italy}
\acmBooktitle{Proceedings of the 48th International ACM SIGIR Conference on Research and Development in Information Retrieval (SIGIR '25), July 13--18, 2025, Padua, Italy}\acmDOI{10.1145/3726302.3730100}
\acmISBN{979-8-4007-1592-1/2025/07}

\usepackage{xcolor}
\usepackage{colortbl}
\usepackage{pdfpages}
\usepackage{array}
\usepackage{multirow}
\usepackage{makecell}
\usepackage{paralist}
\newtheorem{lemma}{Lemma}
\newtheorem{definition}{Definition}

\newcommand{\mb}[1]{\boldsymbol{#1}}

\newcommand{\vx}{\mathbf{x}}
\newcommand{\vy}{\mathbf{y}}
\newcommand{\vd}{\mathbf{d}}
\newcommand{\vq}{\mathbf{q}}

\newcommand{\mD}{\mathbf{D}}
\newcommand{\mQ}{\mathbf{Q}}
\newcommand{\mA}{\mathbf{A}}

\newcommand{\mat}[2]{\left(\begin{array}{#1} #2\end{array}\right)}

\newcommand{\R}{\mathbb{R}}
\newcolumntype{L}[1]{>{\raggedright\arraybackslash}p{#1}}
\newcolumntype{C}[1]{>{\centering\arraybackslash}p{#1}}
\newcolumntype{R}[1]{>{\raggedleft\arraybackslash}p{#1}}

\begin{document}
\title{Towards Lossless Token Pruning in Late-Interaction Retrieval Models}

\author{Yuxuan Zong}
\orcid{0009-0002-0376-1369}
\affiliation{%
  \institution{Sorbonne Université, CNRS, ISIR}
  \city{Paris}
  \country{France}
}
\email{yuxuan.zong@isir.upmc.fr	}

\author{Benjamin Piwowarski}
\affiliation{%
  \institution{CNRS, Sorbonne Université, ISIR}
  \city{Paris}
  \country{France}
}
\email{benjamin.piwowarski@cnrs.fr}

\renewcommand{\shortauthors}{Zong and Piwowarski}

\begin{abstract}

Late interaction neural IR models like ColBERT offer a competitive effectiveness-efficiency trade-off across many benchmarks. However, they require a huge memory space to store the contextual representation for all the document tokens. 
Some works have proposed using either heuristics or statistical-based techniques to prune tokens from each document.
This however doesn't guarantee that the removed tokens have no impact on the retrieval score. Our work uses a principled approach to define how to prune tokens without impacting the score between a document and a query. We introduce three regularization losses, that induce a solution with high pruning ratios, as well as two pruning strategies. We study them experimentally (in and out-domain), showing that we can preserve ColBERT's performance while using only 30\% of the tokens.
\end{abstract}




\maketitle

\section{Introduction} \label{sec:intro}
Pre-trained contextualized language models based on the Transformer architecture, such as BERT, significantly improved retrieval effectiveness over the previous state-of-the-art methods in information retrieval (IR) like e.g. BM25~\cite{BM25}. 
Among these models, dual encoders such as sparse~\cite{SPLADEv2, COIL} or dense~\cite{DPR} models, have the best efficiency while cross-encoders, e.g. MonoBERT~\cite{MonoBERT}, have the best effectiveness.

Another category of models, called \textit{late interaction} models,  like ColBERT~\cite{ColBERT,ColBERTv2}, offer a good efficiency-effectiveness trade-off, as they are halfway between dual encoders and cross-encoders. Specifically, unlike dual encoders, late interaction models use embeddings at the token level rather than to represent the entire query or document.
In ColBERT, the score of a document for one query token is defined as the maximum of the inner product over all the document tokens representations with the query token representation.
This token-level granularity allows for finer semantic matching and improved retrieval accuracy. Moreover, compared to the cross-encoder mechanism like MonoBERT~\cite{MonoBERT}, this "late interaction" mechanism provides a way to build an index of a whole dataset by storing the document token representations, or \textit{ document vectors} in short.

\begin{figure}
    \centering
    \includegraphics[width=0.8\linewidth]{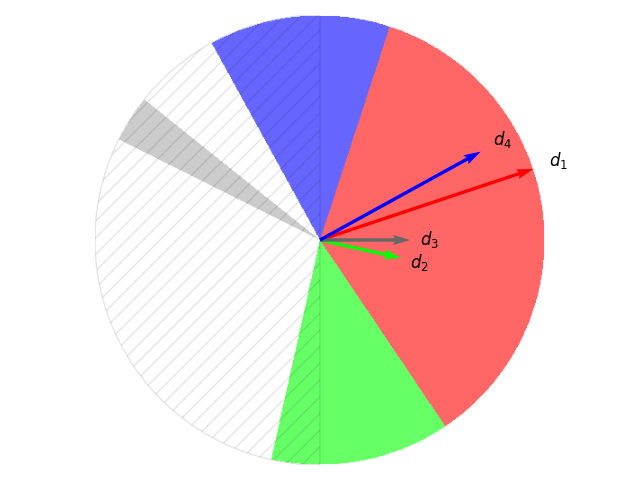}
    \caption{Illustration of the concept of dominance. In this example, $\vd_3$ is dominated by $\vd_1$ and $\vd_2$. $\vd_3$ can be removed from the document representation without changing a modified ColBERT scoring function (see Section~\ref{sec:lp}). Despite its low norm, $\vd_2$ is kept since it brings new information about the relevance of the document. The hashed area corresponds to the half-space where the inner product of any vector with $\vd_3$ is negative. The gray area is discussed in Section~\ref{sec:adapting-colbert}.}
    \label{fig:domi}
\end{figure}

However, storing the vectors of the documents requires a large index space compared to dual encoding approaches. For example, for MS MARCO v1~\cite{MSMARCO}, composed of 8.8M passages, the ColBERT index takes 142 GiB. PLAID~\cite{PLAID}, which is a highly optimized version of ColBERTv2, requires more than 20 GiB of space. This is in contrast to the 13 GiB needed to store the 768-dimensional vectors of a dense model~(without any compression, using FP16). Furthermore, a large index leads to expensive computational costs; as in ColBERT's retrieval pipeline, the retrieval cost scales linearly with the query and document's length. 

To reduce the index size, some works have proposed to prune tokens from the document.
%
Some existing works~\cite{colbert_pruning, static_pruning, late_interaction_pruning} propose to use static methods based on heuristic or statistical approaches, but they degrade ColBERT performance with high pruning ratios, greater than 50\%.
Some other works~\cite{AligneR, ColBERTer} learn to discard some tokens using a gating mechanism, and hence aim at learning sparse (at token level) document representation. These methods remove some of the unimportant tokens, but do not consider how to remove tokens while having a minimum impact on the ColBERT score.

This can be a problem, as illustrated in Figure~\ref{fig:domi}, where the document vector $\vd_2$, while having the lowest norm, can still match a query token -- i.e., the inner product with a query vector located in the green area will be maximized by $\vd_2$. In the same document, $\vd_3$ could, however, be pruned without changing the score of the document, since either $\vd_1$, $\vd_2$ or $\vd_4$ would match any query token located in the blue, red or green areas.


In this paper, we first formally define this \emph{lossless} pruning strategy and show that it is equivalent to a linear programming problem. %
We then propose different training methodologies and approximate pruning strategies that can be used in practice.
Finally, through our experiments, we analyze the performance of our proposed approach considering different pruning strategies and regularizations, leading to the removal of 70\% tokens and a performance drop of less than 1.5\% on in-domain datasets. For the out-of-domain evaluation task, we also end up with a 65\% token removal with less than a 3\% drop in the performance on BEIR task (1.1\% of performance drop in the LoTTE search). We release the experimental code connected to this paper \footnote{\href{https://github.com/yzong12138/MVDR_pruning}{https://github.com/yzong12138/MVDR\_pruning}}.


In the following, we report the research questions that are tackled in this paper, as well as our different contributions.

\paragraph{\textbf{RQ1}} \emph{Is it possible to define a late interaction model and a pruning algorithm that guarantees that the score of a document does not change whatever the query?}
In Section~\ref{sec:dominance}, we formally define the concept of ``dominance'' whereby a document vector can be removed from a document's representation without impacting a late interaction scoring function.  More precisely, in Section~\ref{sec:adapting-colbert}  we discuss how to modify the ColBERT scoring function to allow for more pruning without impacting its performance, and in Section~\ref{sec:lp}, we show how the problem of pruning document tokens \emph{with the guarantee of not modifying the score of the document}, that is, \emph{lossless pruning}, is equivalent to a Linear Programming problem~\cite{nocedalNumericalOptimization2006}.

\paragraph{\textbf{RQ2}} 
\emph{How can this motivate regularization (during training) and practical pruning strategies (before inference), and what is the effectiveness/efficiency trade-off of our overall methodology for late-interaction models?} 
In Section~\ref{sec:methodology}, we propose three regularizations (so that more document tokens are dominated) as well as two pruning strategies that help to increasing the pruning ratio. In Sections~\ref{sec:pruning_ratio} and~\ref{sec:in-domain}, we conduct extensive experiments based on these strategies that show that we can achieve a competitive efficiency-effectiveness trade-off.
 
\paragraph{\textbf{RQ3}}
\emph{How is our proposed regularization and the pruning strategy perform on out-of-domain datasets?} 
As regularization might prevent the model from generalizing to out-of-domain datasets, we conduct extensive experiments on a subset of the BEIR and LoTTE datasets in Section~\ref{sec:OoD}. Although the difference with non-pruned ColBERT is greater than in the in-domain scenario, we still outperform other pruning approaches and achieve a good efficiency-effectiveness trade-off.

\paragraph{\textbf{RQ4}}
\emph{Is it better to train from an MLM or from an IR checkpoint?}
As the proposed regularizations have an impact on the geometry of space, we question whether starting with an IR checkpoint, i.e. ColBERTv2, is better suited for pruning than starting from an MLM one. In Section~\ref{sec:Scratch}, we show that the impact is limited.

\section{Related Work} \label{sec:related_work}

\paragraph{Multi-vector dense retrieval} 

Multi-vector dense retrieval (MVDR) can be seen as a generalization of single-vector dual encoder models. Instead of encoding the full content of a query or document into a single vector of the representation space, MVDR uses fine-grained token-level modeling for scoring. 
Neural IR models using this type of strategy are called interaction models. The first models were based on non-contextualized token (or word) embeddings, and relied on complex post-processing of the interaction between document and query vectors to accommodate the lack of token contextual encoding~\cite{DRMM, K-NRM, Co-PACRR}. 

With transformer-based architectures, since token representations are contextualized, 
MVDR models rely on a much simpler interaction mechanism.
Models such as ColBERT models~\cite{ColBERT, ColBERTv2}, XTR~\cite{xtr}, COIL~\cite{COIL, COILcr}, ME-BERT~\cite{ME-BERT}, MVR~\cite{MVR}, and CITADEL~\cite{citadel} compute query-document relevance by computing, for each query token, the highest inner product of its representation with that of a document token representation -- see Eq. \eqref{eq:colbert} in Section~\ref{sec:dominance}. 
Some other works like AligneR~\cite{AligneR}, MVR~\cite{MVR}, ME-BERT~\cite{ME-BERT} propose that a query token could match multiple document tokens in the scoring stage. Although this improves the results, the use of only one token (argmax) is competitive and simplifies the retrieval architecture. We thus stick to the maximum inner product (unary salience) in this work.

\paragraph{Reducing the index size in MVDR models} 

Although MVDR models are more powerful than dual encoders, this, however, comes at a cost during inference. Late interaction methods~\cite{ColBERTv2, AligneR} are based on a two-stage (or more) retrieval, where the first candidate documents are retrieved using vector indices such as Faiss~\cite{faiss}. To avoid having too many vectors in the index, some works have proposed to first cluster them and index only the cluster centroids~\cite{PLAID, ColBERTv2, EMVB} -- which amounts to selecting candidate documents before scoring them by loading the full document representations. To avoid this re-ranking step, some approaches like SLIM and MUVERA~\cite{slim, MUVERA} propose to directly rely on cluster-based representations, which avoids the complex multistage retrieval and increases the inference speed -- but at the cost of decreasing the performance.


A complementary method to reduce index size is to directly prune document tokens before storing them~\cite{colbert_pruning, static_pruning, late_interaction_pruning, ColBERTer, AligneR}. 
Some of these methods are based on heuristics, such as pruning based on stopwords~\cite{static_pruning}, IDF value~\cite{static_pruning, late_interaction_pruning}, attention score~\cite{colbert_pruning, late_interaction_pruning}, or position~\cite{colbert_pruning, late_interaction_pruning}. Specifically, for attention score-based methods~\cite{late_interaction_pruning}, a token is important if many tokens are similar to it but not the other way around (i.e., tokens whose representation is close to a centroid of clusters of document vectors).
These attention-score methods bear some similarity with our model, but are not based on any training and might still discard important tokens. 
These \emph{static} strategies (the Transformer model is kept unchanged) generally lead to a low efficiency-effectiveness trade-off.
For example, repeated tokens such as ``doctor'' have the same IDF value, but keeping all of them might not provide supplementary information to the scoring function.
%

To further push the efficiency-effectiveness trade-off, some works propose \emph{dynamic} strategies in which the model learns which tokens should be pruned. AligneR~\cite{AligneR} employs entropy-regularized optimization (relaxed top-$k$ optimization), and ColBERTer~\cite{ColBERTer} reduces the number of vectors by combining sub-word vectors into unique whole-word representations and learns to discard them through a trainable gate vector. Similarly, He et al.~\cite{LEAP_MV} propose to learn if a vector should be ``keep'' or ``drop'' using Gumbel-Softmax sampling. Const-BERT~\cite{Const-BERT} proposes learning a pooling layer that projects document vectors to a smaller set of vectors (constant cardinality).
Our work follows this line of research, but we explore more principled ways to prune tokens, which in theory can lead to lossless pruning strategies and which work better in practice.

\section{Pruning ColBERT – Theory} \label{sec:dominance}

In this section, we first show how to adapt ColBERT so that more tokens can be pruned and define the concept of \emph{dominance} as well as its equivalence with a linear programming problem.

\subsection{Adapting ColBERT for Pruning}
\label{sec:adapting-colbert}

ColBERT uses an encoder to represent document and query tokens separately, before computing a score based on the inner product between query and document token representations. More formally, given a document $\mD$ and a query $\mQ$, each document (resp. query) token is represented by a vector $\vd \in \mD$ (resp. $\vq \in \mQ$) of dimension $d$. The score of a document $\mD$ for a given query $\mQ$ is given by the following formula: 
\begin{equation}
    \mathrm{ColBERT}(\mQ, \mD) = \sum_{\vq \in \mQ} \max_{\vd \in \mD} \vq \cdot \vd
    \label{eq:colbert}
\end{equation}
In this equation, as well as in the remainder of the paper, $\mD$ denotes both a matrix (where rows are the dimension of the latent space and columns the tokens of the document) \textit{and} a set of vectors in the latent space -- given the max operator of the scoring function, the two formulations are in practice equivalent. We also shorten the \textit{phrase document (resp. query) token representation} by \emph{document (resp. query) vector}.

To prune tokens, we need to find a subset $\mD^+=\mD \setminus{\mD^-}\subseteq \mD$ such that
$$
\forall \vq\in\R^d, \ \max_{\vd \in \mD} \ \vq \cdot \vd = \max_{\vd \in \mD^+} \vq \cdot \vd
$$

In ColBERT, the query and the document vectors have unit norm ($\|\vq\|=\|\vd\|=1$), that is, the model implicitly computes the cosine between the vectors $\vq$ and $\vd$. 
In this setting, unless two document tokens have \textit{exactly} the same representation, 
there is no way to remove any token $\vd^-$ from $D$ without impacting this maximum
-- since with $\vq = \vd^-$, we have $\vq\cdot \vd^- > \vq \cdot \vd$ for any $\vd\in \mD \setminus \{\vd^-\}$.

To allow for more tokens to be removed, we propose to modify this ColBERT scoring function by the following: 
\begin{equation}
    \mathrm{ColBERT_{P}}(\mQ, \mD) = \sum_{\vq \in \mQ} \max_{\vd \in \mD} [\pi_\theta(\vq) \cdot \pi_\theta(\vd)]_+
    \label{eq:colbert:relu}
\end{equation}
where $\pi_\theta$ is a projection operator and $[x]_+ = \mathop{ReLU}(x) = \max(0, x)$. 
To simplify the notation, in the remainder of the paper, we overload $\vq$ and $\vd$ as representations \emph{after} the projection $\pi_\theta$.

Compared to Eq.~\eqref{eq:colbert}, the first modification is to use a projection $\pi_\theta$ of the document/query tokens. 
This ensures that their norm be less or equal to 1, so that some tokens can be pruned on the document side. On the query side, this might help the model cope with more or less important tokens, but this is not strictly required for pruning. 
The second modification we make is to only consider positive inner products -- which in practice is already most of the time true in ColBERT because of the anisotropy of the representation space~\cite{godeyAnisotropyInherentSelfAttention2024}. 

Both of these changes are needed to remove some tokens from the document without affecting the ColBERT$_P$ score. This is illustrated in Figure~\ref{fig:domi}, where the document token that reaches the maximum inner product with a query $\vq$ is indicated by a color code, and where the inner product with $\vd_3$ is negative for any query located within the half-space indicated by the hatched area. In the other half-space, $\vd_1$, $\vd_2$ or $\vd_4$ have a higher inner product with any query vector $\vq$ than $\vd_3$. Note that without the ReLU transformation, $\vd_3$ could have a maximum inner product (gray zone) below zero, since the magnitude of the inner product with $\vd_3$ is less than with other document vectors in this area.


AligneR~\cite{AligneR} and ColBERTer~\cite{ColBERTer} do use another way to reduce the norm of a document vector, using a gating mechanism by predicting a scalar in $[0,1]$ that is multiplied by the document vector at hand.
Our projector $\pi_\theta$ reduces not only the vector norm and adds the ability to discard document vectors. This is important, especially when the norm is reduced but not to (a numerical) zero.
Returning to Figure~\ref{fig:domi}, according to the pruning scenario based on the norm (e.g., salience in the notation of AligneR) of the vectors~\cite{AligneR, ColBERTer},  $\vd_2$ should be pruned instead of $\vd_3$. With our approach, $\vd_3$ can be pruned without impacting the scoring function.





\subsection{Dominance and Linear Programming}
\label{sec:lp}

\paragraph{Local and Global Dominance}

Based on Eq.~\eqref{eq:colbert:relu}, we propose the following definition that allows removing tokens from a document without changing the $\mathrm{ColBERT}_P$ score.
\begin{definition}[Local dominance]
A vector $\mb d^- \in \mathbb R^d$ is \emph{dominated} by a set of vectors $\mb D$, which we denote by $\mb d^- \prec \mb D$, if and only if for any $\mb q\in\R^d$, 
\begin{equation}
\label{equ:local_domi}
\mbox{ either } \mb q\cdot \mb d^- \le 0 \mbox{ or }
\exists \mb d^+\in \mb D \mbox{ such that } \mb q\cdot \mb d^+ > \mb q \cdot \mb d^-
\end{equation}
\end{definition}
In this case, it is is easy to show that for any query vector $\mb q$,
$$
\max_{\mb d\in \mb D\cup \{\mb d^-\}} [\vq \cdot \vd]_+ = \max_{\mb d \in \mb D \setminus \{\mb d^-\}} [\mb q \cdot \mb d]_+
$$

Based on this definition, we can partition a set $\mb D$ of vectors into $\mb D^+$ and $\mb D^-$, by defining $\mD^-$ as the set of all the vectors $\vd$ dominated by $\mb D$.
The question remains whether any vector in $\mb D^-$ is dominated by $\mb D^+$.
If this is true, we only need to test for \emph{local dominance} to define the whole set of vectors that can be pruned without impacting the ColBERT$_P$ score. The following lemma proves this.

\begin{lemma}[Global dominance]
\label{lemma:global}
 Let $\mD^+$ and $\mD^-$ be a partition of the set of document vectors $\mD\subset\R^d$ such that any document vector $\vd \in \mD^-$ is dominated by $\mD$. Formally, 
 $$
 \mD^- = \left\{ \vd \in \mD \middle| \vd \prec \mD \right\} 
 \mbox{ and }
 \mD^+ = \mD \setminus \mD^-
 $$
 
 Then, for any document $\vd^- \in \mD^-$, $\vd^-$ is dominated by $\mD^+$. 
\end{lemma}


\begin{proof} \label{prf:glb_domi}
To prove this lemma, we choose an arbitrary $\vq\in\R^d$ and $\vd_1 \in\mD^-$, and show that there exists $\vd^+\in\mD^+$ such that $\vq \cdot \vd^+ > \vq \cdot \vd_1$.

First, for for any $\vd_k\in \mD^-$, let us denote $\mD_\vq(\vd_k)$ a subset of $\mD^-$ such that 
any document vector in $\mD_\vq(\vd_k)$ have a higher inner product with $\vq$ than $\vd_k$, i.e.
$$\mD_\vq(\vd_k) = \left\{  \vd\in \mD^- \middle| \ \vd\cdot \vq > \vd_k\cdot \vq \right\}$$

If for some $k$, $\mD_\vq(\vd_k)\not=\emptyset$, we can choose $\vd_{k+1}\in \mD_\vq(\vd_k)$. By construction, $\mD_\vq(\vd_{k+1}) \subset \mD_\vq(\vd_k) \subset \mD^-$ and $\vq \cdot \vd_{k+1} > \vq \cdot \vd_{k}$.  As $\mD^-$ is finite, by induction we can conclude that for some $k^*$ we have $\mD_\vq(\vd_{k^*})=\emptyset$. 

Since $\vd_{k^*} \prec \mD$, this implies that 
$\exists \vd^+\in \mD\setminus \mD^- = \mD^+$ such that $\vd^+ \cdot \vq > \vd_{k^*} \cdot \vq$.
This in turn implies that there exists $\vd^+\in\mD^+$ such that 
$$
\vq \cdot \vd^+ > \vq \cdot \vd_{k^*} > \ldots > \vq \cdot \vd_1
$$
Since this holds for any query vector $\vq$, this implies that $\vd_1 \prec \mD^+$, which concludes the proof.
\end{proof}

\paragraph{Exact pruning and Linear Programming}

The above lemma implies that we only need to test local dominance to construct $\mD^+$, i.e. to find the vectors in $\mD$ that verify Eq.~\eqref{equ:local_domi}. We show that this corresponds to a Linear Programming (LP) problem below by using Farka's lemma~\cite{Farkas}.

%
%
Let us first state the version of Farka's lemma that we use.

\begin{lemma} [Farkas' Lemma]
Let $\mb A \in \R^{d\times n}$ and $\mathbf{b} \in \R^{m}$, then exactly one of the following assertion is true (the vector inequalities mean that they must hold for any component of the vector): 
\begin{enumerate}
    \item There exists $\mb{x} \in \R^n$ such that $\mb A\mb x = \mb{b}$ and $\vx\geq 0$.
    \item There exists $\vy \in \R^{d}$ such that $\mb A^\top \mb \vy\geq 0$ and $\mb{b}^\top\vy\leq 0$.
\end{enumerate}
\end{lemma}

Setting $\mA = \mat{ccc}{\vd - \vd_1 & \ldots & \vd - \vd_n}$ and $\mb{b} = -\vd$, where $\vd$ is our target document vector, we can use the lemma to show that if
\begin{equation}
    \exists \mb{x} \in \R^n,
    \mat{ccc}{\vd - \vd_1 & \ldots & \vd - \vd_n}^\top\mb x 
= -\mb{d} \mbox{ and } \vx\geq 0
\label{eq:lp-problem}
\end{equation}
then for any query vector $\vq$ either $\vd^\top\vq\le 0$, or there exists $i$ such that 
$$
\left(\vd - \vd_i\right)^\top \vq < 0 \Leftrightarrow \vd \cdot \vq < \vd_i \cdot \vq
$$
which corresponds exactly to the definition of local dominance.




We therefore have shown that the problem of determining whether a document token is dominated is equivalent to a linear programming (LP) problem which can be solved using linear programming solvers~\cite{nocedalNumericalOptimization2006}. 
However, LP problems are computationally expensive to solve. The complexity of the simplex algorithm is not guaranteed to be polynomial in the worst case~\cite{simplex}, and there is still ongoing work to improve the time complexity of other approaches, such as the methods of interior points~\cite{LP_solve_1, LP_solve_2}. 

To reduce the computational cost, we describe below simple techniques to determine whether $\vd$ belongs to $\mD^+$ or $\mD^-$, before using LP solvers for the remaining ones.

First, according to the lemma~\ref{lemma:global} that links local to global dominance, if $\vd^-$ is dominated, then whatever the query vector $\vq$, we can find $\vd^-\in\mD^+$ that dominates it, that is, for which $\vq \cdot \vd^+ > \vq \cdot \vd^-$. This shows that whenever we find a dominated vector $\vd^-$, we can directly remove it from $\mD$  for the remaining local dominance tests. 

Second, we compute the inner product of each document vector $\vd$ with the others and look at whether the vector matches itself, that is, whether $\vd \cdot \vd \ge \vd \cdot \vd^\prime$ for $\vd\in\mD$. If this is the case, we know for sure that $\vd \in \mD^+$. In practice, between 20 and 30\% of ``dominating'' tokens can be identified by this simple strategy.

\section{Methodology} \label{sec:methodology}

In this section, we describe a methodology for training and pruning in practice our ColBERT$_P$ model, based on the dominance theory we described in the previous section.

\subsection{Approximate Pruning}

The LP-based process described in Section~\ref{sec:lp} is exact but has two main limitations:
\begin{itemize}
    \item Some vectors might have a very small component that prevents them from being dominated; 
    \item Solving an LP problem is time-consuming.
\end{itemize}
To tackle the first problem, we propose using dimensionality reduction; for the latter problem, we propose using a very simple heuristic, based on the norm of the vectors, and we experimentally study how well it approximates the LP solution.

\subsubsection{LP-based Pruning with Reduced Dimension}
\label{sec:lp:reduce}

Based on the dominance analysis in the last section, a document token $\vd \in \mD$ should be pruned if, for any query vector $\vq$, we can find $\vd^+ \in \mD$ such that $\vd^+ \cdot \vq > \vd \cdot \vq$. Then a requirement for local dominance is that $\vd$ is in the span of vectors in $\mD\setminus\{\vd\}$. 
However, this rarely happens in practice because of the MLM pre-training that induces a strong prior on the model, that of trying to preserve unique token semantic information. 

We hence propose to project the vectors in a lower-dimensional space by using truncated singular value decomposition (SVD). Using a truncated SVD allows for preserving dimensions which are the most common between the vectors, and thus discarding the vector residuals that prevent vectors from belonging to the same subspace. 

Formally, given the reduced singular value decomposition of our input document $\mD = \mb U \mb{\Sigma}\mb V^\top$, with 
$\mb U\in\R^{d \times m}$, $\mb \Sigma\in \R^{m\times m}$ a diagonal matrix with decreasing and non-negative diagonal components, and $\mb V\in \R^{m \times n}$ with $m \le \min(d,n)$,
we can write a document vector $\vd_i$ as $\mb U \mb \Sigma \mb v_i$ with $\mb v_i$ the column $i$ of $\mb V^\top$. For any query $\vq$, we have 
$$\vq \cdot \vd_i = \vq ^\top \mb U \mb \Sigma \mb v_i = (\mb U^\top \mb q)^\top \left(\mb \Sigma\mb v_i \right)$$ 
The above equation shows that the local dominance problem can be solved within a space of dimension $m \le d$, by posing $\vq^\prime=\mb U^\top \vq \in \R^{m}$ since $\vq$ is arbitrary. This has an impact on matrix conditioning, which helps stabilize the LP algorithm~\cite{Condition_number_for_lp_1, Condition_number_for_lp_2, Condition_number_for_lp_3}, and also allows to speed up the overall local dominance decisions. Formally, the problem is equivalent to
$$
\bar \mD = \mb\Sigma \mat{ccc}{\mb v_1 & \ldots & \mb v_n }
$$

We can push this further by making the decision of local dominance \emph{approximate}.
Denoting $\mb \Sigma_k$ the first $k<m$ columns of $\mb \Sigma$ and $\mb v_i^{(k)}$ the first $k$ components of $\mb v_i$, we know that $\mb U \Sigma_k \mb v_i^{(k)}$ is the best approximation of $\vd_i$ in a $m$-dimensional subspace (in terms of Frobenius distance). 
We therefore propose to use $\bar\mD^{(k)}=\mb \Sigma^{(k)}\mat{ccc}{\mb v^{(k)}_1 &\ldots &\mb v^{(k)}_n}$ instead of $\mD$ as the document matrix, which helps to remove the barely used dimensions (and thus increases the number of dominated vectors).

In practice, we use $k$ such that most of the singular values are covered by $\Sigma_k$, that is, $$\frac{tr(\mb\Sigma)-tr(\mb\Sigma^{(k)})}{tr(\mb \Sigma)} \ge \theta_{LP}$$ where $\theta_{LP}$ is an arbitrary bound. Note that dimensionality dimension is only used to prune the vectors -- i.e. we use the non-pruned vectors $\vd\in\mD$ to represent a document.

\subsubsection{Norm Pruning}

A lightweight alternative to LP-based pruning is to rely on norm-based pruning, such as in AligneR~\cite{AligneR} or ColBERTer~\cite{ColBERTer}. 
As having more dominated vectors implies having more low norm vectors, it is natural to question whether using a norm-based strategy for pruning is interesting in our case.
Compared to the other methods, the pruning based on the norms is fast and the impact on the score is directly related to the norm of the document vector. 
In our experiment, we denote the $\theta_{N}$ as the threshold for this pruning strategy, where it removes all document tokens that have a norm smaller than this threshold.

However, the fact that other document vectors might be closely related -- or not -- to a document vector with a low norm is not taken into account. This is illustrated in Figure \ref{fig:domi}, where $\vd_3$ is not pruned as it should and $\vd_2$ is pruned but might have an impact on the retrieval score. 

\subsection{Regularization}
\label{sec:regularization}

Using an already trained ColBERT model and pruning the document vectors as described above might not work well because the model is not inherently biased toward producing document representations that are prone to high pruning ratios -- even if not using the cosine product. 
As the pruning strategies that rely on the LP problem are not easily differentiable, we propose three different regularizations that can be added to the IR-related loss so that the resulting set of document vectors can be pruned more effectively. 

\paragraph{Nuclear norm regularization}

To force some more document tokens to dominate, we need to ensure that the document vectors are in the span of each other. Given the matrix $\mD$, this is linked to the problem of minimizing the nuclear norm of $\mD$, which is formally defined as the $L_1$ norm of the singular values of the matrix, i.e., $\mathrm{trace}\left(\mb \Sigma\right)$ using the notation defined in Section~\ref{sec:lp:reduce}. As $L_1$ pushes some components to zero, this reduces the size of the subspace where the document vectors in $\mD$ lie, and also puts pressure on reducing the norm of the document vectors. Both are linked to a higher pruning ratio.
Formally, our regularization is defined as
\begin{equation}
\mathcal{L}^{(\mathrm{nuc})}_{r} = \frac{1}{\min(n, d)}\|\mD\|_{nuc}
\end{equation}
where $n$ is the number of document vectors and $d$ the dimension of the token representation space.





\paragraph{Document tokens similarity matrix regularization}

Although the nuclear norm pushes the document vectors to a small-dimensional subspace, it has two problems. First, its computation is time-consuming. Second, it does not really ensure that some document vectors are dominated.

We therefore propose to rely on a simpler measure based on the inner product between $\vd$ and the other vectors $\vd^\prime$ of $\mD$, whereby we want to maximize the inner product $\vd \cdot \vd^\prime$ between $\vd$ and $\vd^\prime$. We then modify this inner product to match our needs:
\begin{itemize}
    \item  As we don't want the norm of $\vd$ to increase (which would reduce its likelihood to be pruned), we divide the inner product by the norm of $\vd$;
    \item We only consider \emph{positive} inner products to avoid the norm of $\vd^\prime$ to be reduced when maximizing the inner product;
    \item Finally, we want more pressure on low norm document vectors $\vd$;
\end{itemize}

Putting all together, this leads to the following regularization (we take the opposite since we frame it as a minimization problem):
\begin{equation}
\mathcal{L}^{(\mathrm{sim})}_{r} = - \frac{1}{n(n-1)}\sum_{\vd \in \mD} (1 - \|\vd\|_2) \sum_{\vd^\prime \in \mD \setminus \{\vd\}} \frac{[\vd \cdot \vd^\prime]_+}{\|\vd\|_2 + \varepsilon}
\end{equation}
where $n$ is the number of document vectors and $\varepsilon$ is a small value to avoid numerical instability when calculating the gradients.\footnote{We use the value $\varepsilon=0.01$ in our experiments.} 

\paragraph{L1 norm}
Similarly to ColBERTer~\cite{ColBERTer}, which uses an L1 norm regularization of sigmoid-based gates (at word level), we propose to use L1-norm regularization that helps to reduce the number of nonzero components in vectors. This in turn allows us to prune vectors, since if $\|\vd\|=0$, the document vector $\vd$ is straightforwardly dominated. Our third regularization loss is thus:

\begin{equation}
\mathcal{L}^{(L1)}_{r} = \frac{1}{n}\sum_{\vd \in \mD} \|\vd\|_1
\end{equation}

\subsection{IR loss}

Our main loss is the one used ColBERTv2~\cite{ColBERTv2} which ensures that we learn an IR scoring function. ColBERTv2 uses both a distillation and an infoNCE loss.
We use the hard negatives provided by a cross-encoder.\footnote{We use the  \href{https://huggingface.co/colbert-ir/colbertv2.0_msmarco_64way/resolve/main/examples.json?download=true}{same}  datasets as as ColBERTv2} 
Based on this collection, we use tuples consisting of one query $q$, one annotated positive passage $p^+$ and a hard negative $p^-$, as well as a set of easy negative passages $P^-_{easy}$ (in-batch negatives). 
Following previous work~\cite{ColBERTv2, RocketQAv2, tas-balanced, citadel, tct-colbert, distillation}, we use a KL-divergence loss to distill the teacher's score into the student architecture.
Using the following estimation of the probability of $p_+$ being a positive passage,
$$p(p_+|q, N) = \frac{e^{s(q, p_+)}}{e^{s(q, p_+)} + \sum_{p \in N} e^{s(q, N)}},$$
the IR loss is defined as a combination of distillation and cross-entropy:
\begin{multline}
\mathcal{L}_{IR} = KL(p_{ColBERT_+}(. | q, N_h) || p_{teacher}(.|q, N_h)) \\ + \log p_{ColBERT_+}(p_+|q, N_a)
\label{eq:ir-loss}
\end{multline}
with $N_a$ all the documents (positive, hard, and easy negatives), and $N_h$ the positive and hard negative pair.


\subsection{Overall}

The overall loss of our training is the combination of the regularization and IR losses, defined as follows:

\begin{equation}
\mathcal{L}_{final} = \mathcal{L}_{IR} + \alpha \mathcal{L}_{r}
\end{equation}
where $\alpha$ corresponds to the coefficient of the regularization and $\mathcal{L}_r$ is one of the regularization losses defined in Section~\ref{sec:regularization}.


\section{Experiments} \label{sec:xps}
In this section, we describe the experiments we develop to validate our pruning strategies and the proposed regularizations.

\subsection{Experiment Setup}

\paragraph{Datasets} 

We train our model on the MS MARCO v1 passage retrieval dataset~\cite{MSMARCO}. This dataset contains 8.8M passages. 
We evaluate the model on the MS MARCO dev set (6980 queries), as well as the TREC 2019 DL (43 assessed topics) 
and TREC 2020 DL (54 assessed topics) sets. 
%
We also report the zero-shot performance of our training and pruning strategies using a subset of the BEIR~\cite{BEIR} benchmark, as well as of the LoTTE collections~\cite{ColBERTv2}~(datasets were chosen to lower the overall computation budget).

\paragraph{Pipeline} \label{sec:pipeline}

We evaluate our model using a reranking-based pipeline. 
Specifically, we use our own trained SPLADEv2~\cite{SPLADEv2} as a first-stage retrieval and
then use our model to rerank the top-k document. 
We did not use the standard ColBERT pipeline in our experiments for two reasons. 
First, following Formal et al.~\cite{SPLATE}, 
the SPLADE model is a very effective first-stage retriever that is competitive with the first-stage used in ColBERT, namely PLAID~\cite{PLAID} or EMVB~\cite{EMVB}. 
Second, in our model, as document tokens are not normalized, the centroid interaction used by PLAID and EMVB~\cite{PLAID, EMVB} does not work.
The document tokens with small norms are more likely to be in the same cluster, and this increases the number of candidates to be re-ranked, which is problematic. 
Finally, as we experiment with different ColBERT variations and pruning strategies, building a PLAID~\cite{PLAID} or EMVB~\cite{EMVB} index would be quite costly.

For the evaluation, we use the top 100 passages on the MS MARCO dev set, BEIR, and LoTTE (given the high number of topics), and the top 1000 for the evaluation on the TREC DL 2019 and 2020.

\subsection{Implementation Details}

We conduct experiments\footnote{PyTorch 2.0.1,  HuggingFace transformers 4.37, XPM-IR 1.3.1~\cite{Xpmir}} with NVIDIA TESLA V100 gpus. 
We use an AdamW optimizer with a learning rate of 1e-5, a batch size of 16, and train during 320k steps. We use the ColBERTv2 checkpoint to initialize our model parameters\footnote{\href{https://huggingface.co/colbert-ir/colbertv2.0}{https://huggingface.co/colbert-ir/colbertv2.0}}~\cite{ColBERTv2}.
We validate the training result on a subset of the MS MARCO development set (1000 topics, distinct from the final dev-small evaluation set) every 8000 steps. 
To select the checkpoint providing a good trade-off between effectiveness and efficiency, as a validation metric, we use a weighted harmonic mean, i.e. H0.5, H1 and H2, using MRR@10 and an estimation of the pruning ratio.\footnote{On 1024 documents, using approximate LP-based pruning with a cumulative singular value proportion of 0.7}
The training takes around 5 days (L1 and document similarity regularization) or 8 days (nuclear norm regularization) on a single GPU.

For the projection $\pi_\theta$, different from previous work~\cite{AligneR, ColBERTer} which computes a scalar in $[0,1]$ using a linear layer followed by a sigmoid, we propose to project the last hidden layer of the BERT model (denoted as $\vd^{BERT}$) onto a vector of dimension $d^\prime > d$. We normalize this vector so it is norm 1, before restricting it to the first $d$ dimensions. In practice, we obtain a vector of norm equal to or less than 1, as well as learning to remove the not-important directions of the original document tokens. Formally, we have 
$$
\pi_\theta(\vd^{BERT}) = \textrm{normalize}\left(\mat{c}{\mb W_1\\ \mb W_2}\vd^{BERT}\right)_{:d}
$$
where $\mb W_1$ is the ColBERT projection matrix (so we can use the ColBERTv2 checkpoint to initialize $\mb W_1$), and $\mathbf{W}_2$ is a projection matrix (with an output dimension of 32) which is initialized randomly.

\paragraph{Baseline}

We first consider the following baselines:
\begin{compactitem}
    \item A standard IR baseline, BM25~\cite{BM25},
    \item SPLADEv2~\cite{SPLADEv2}, our first-stage retriever, and a state-of-the-art sparse model;\footnote{We did not use SPLADEv3~\cite{Spladev3} since its training strategy is more complex than the one used in this paper}
    \item ColBERTv2~\cite{ColBERTv2} end-to-end which uses the full ColBERTv2 pipeline;
    \item ColBERTv2~\cite{ColBERTv2} re-ranking, which re-reranks the documents retrieved by SPLADEv2;
\end{compactitem}

Our second set of baselines are late-interaction models that prune document vectors -- note that all the static methods in the pruning part are based on ColBERTv2~\cite{ColBERTv2}:
\begin{compactitem}
    \item AligneR~\cite{AligneR} is a state-of-the-art ``learning-to-prune'' method that employs an entropy-regularized loss to remove tokens from a document representation. We use only their reported results based on the ``base'' size of BERT~\cite{BERT} to make the result comparable with ours\footnote{AligneR is trained without distillation}; 
    \item Const-BERT~\cite{Const-BERT} projects the document vectors into a fixed number of vectors;
    \item \emph{IDF} pruning~\cite{late_interaction_pruning} that removes tokens with the lowest (token-based) IDF. We copy their results with 50\% pruning ratio in our table;
    \item \emph{First-$p$}~\cite{late_interaction_pruning} that only keeps the first $p$ tokens. We copy their results with 50\% pruning ratio in our table;
    \item \emph{Stopwords} pruning~\cite{static_pruning} that removes tokens from a pre-determined list of stopwords; 
    \item \emph{Attention}-based pruning~\cite{late_interaction_pruning} (see Section~\ref{sec:related_work});
    \item \emph{ColBERTer}~\cite{ColBERTer}. It is a DistilBERT-based ColBERT model in which tokens are first aggregated into words before being pruned.
\end{compactitem}

Among all the results reported in these papers, we did select the ones reported in the tables by either using the configuration put forward by the authors (AligneR, ColBERTer, Const-BERT), or selecting the best trade-off between pruning ratio and effectiveness (for IDF, first-$p$, stop-words). 


\section{Results and Analysis} \label{sec:res}

In this section, we perform some quantitative and qualitative analysis based on our experimental results.


 
\subsection{(In-Domain) Pruning Ratio Analysis} \label{sec:pruning_ratio}
\begin{figure}
    \centering
    \includegraphics[width=1\linewidth]{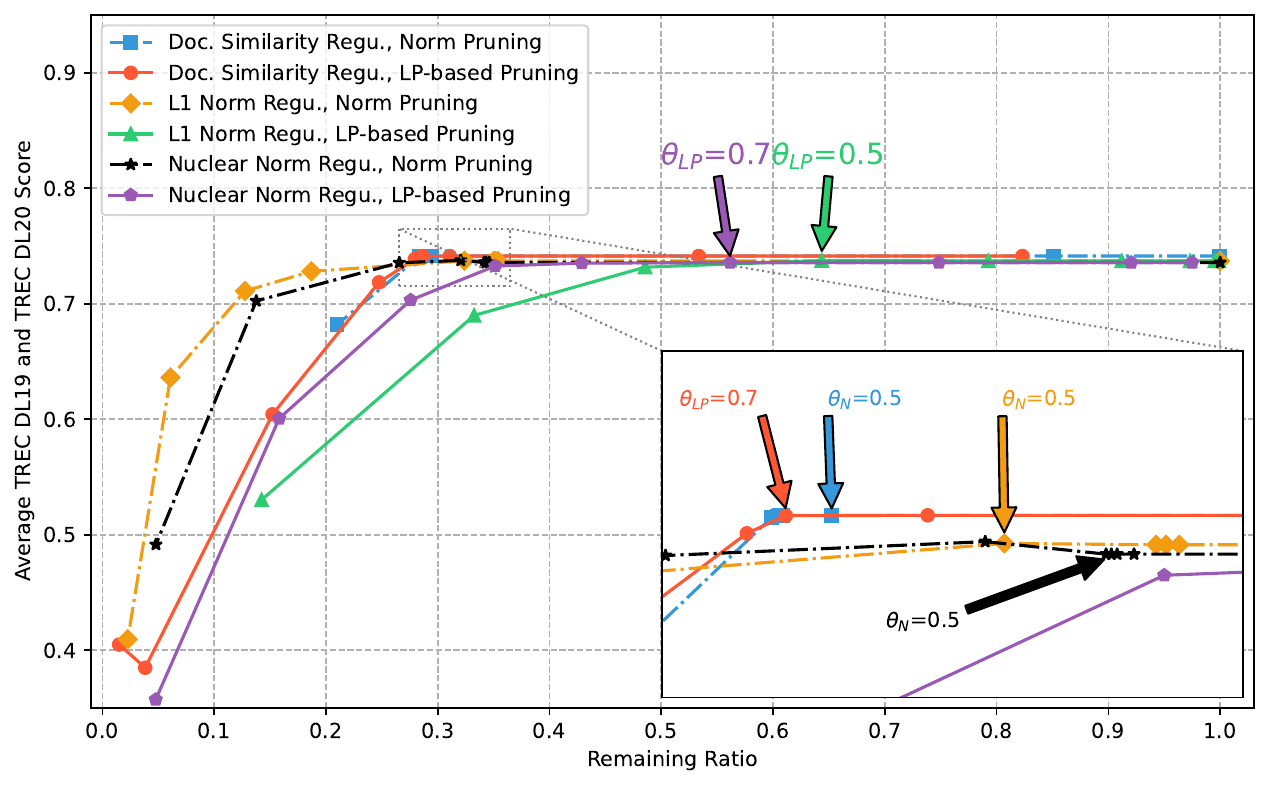}
    \caption{Average TREC DL (2019 and 2020) nDCG@10 for different regularizations and pruning ratios}
    \label{fig:pruning_analysis}
\end{figure}
We first examine the impact of varying pruning thresholds on the pruning ratio and evaluation results, based on our Norm (threshold from 0 to .999) and LP-based (threshold between 0.2 and 0.95) pruning strategy. We keep the regularization weight constant for each of our strategies ($\alpha=0.01$ for the L1 norm and $\alpha=0.8$ for Document Similarity and $\alpha=0.1$ for nuclear norm). These weights were set so that they perform similar without any pruning, to measure the effect of the pruning algorithm itself.

%
In Figure \ref{fig:pruning_analysis}, we can observe that, in general, all our models have no clear performance drop until around 30\% to 40\% of the remaining tokens. Beyond this threshold, we can observe some clear drop in performance. 
The best performance is achieved using document similarity-based regularization, with LP or norm-based pruning.
If we want to prune more, the best models are based on L1-Norm or nuclear norm regularizations, with norm pruning. 
Comparing the different pruning strategies, we found that $\mathcal{L}_r^{(sim)}$ is the regularization that works best for the LP pruning strategy. It provides a similar pruning ratio to the models with the L1-norm pruning, which means that our designed regularization is close to the idea of our proposed theory based on the LP. In contrast, $\mathcal{L}_r^{(nuc)}$ and $\mathcal{L}_r^{(L1)}$ lead to a larger difference in the pruning ratio between the two pruning strategies. It shows that these two methods' approximations are far from our dominance theory and LP procedure.

\subsection{In-Domain Retrieval} \label{sec:in-domain}


In table \ref{tab:in-domain}, we show our in-domain evaluation result, compared with baselines. 
Based on the results presented in the previous section, we use $\theta_{LP}=0.7$ and $\theta_{N}=0.5$.

For each trained model, to maintain a good trade-off between the evaluation result and the pruning ratio, we use the checkpoint with the best weighted harmonic value H2 to evaluate our final model, with more emphasis on the pruning ratio.

Comparing our methods with other state-of-the-art methods, we can observe that our best model (similarity loss $\mathcal L_r^{(sim)}$ with $\alpha=0.8$ and LP pruning with $\theta_{LP}=0.7$) is on par with the ColBERTv2 model on the dev-small set~(MRR@10 of 39.7 vs. 40.0) and the set DL19 (nDCG@10 of 74.8 vs. 74.4), while only keeping 40\% of the total document tokens. The performance is a bit worse on DL20 (nDCG@10 73.3 vs 75.6) but still well over the SPLADEv2 original ranking (68.7).

Among our models, all regularizations have a direct impact on pruning, since increasing $\alpha$ leads to higher pruning rates. We observe that L1-norm pruning is adapted to nuclear and L1-norm regularization, achieving similar performance with a higher pruning ratio. However, for our similarity-based regularization, LP-based pruning is more advantageous, showing that this regularization is more related to dominance than the others.

Finally, comparing our methods with the other state-of-the-art methods which use the static (e.g. IDF, Attention Score, Stopwords, First $p$) or dynamic (e.g. AligneR) pruning methods, all our regularization outperform their result on both the evaluation effectiveness and the pruning ratio. 


%



\begin{table}[h!]
\renewcommand\cellalign{l}
\caption{The in-domain evaluation results (MRR@10 for dev small, nDCG@10 for TREC DL19 and DL20), and remaining ratio (for pruning models). We highlight the models and corresponding pruning strategy we use in our further experiments (bold). $\dagger$: NO statistical significant performance drop with  ($p\leq0.05$) when comparing with ColBERTv2 re-ranking with the two-tailed Student’s t-test.} \label{tab:in-domain}
\begin{tabular}
{m{1cm}m{1.25cm}C{1cm}C{1cm}C{1cm}C{1cm}}
\Xhline{1pt}
\multicolumn{2}{c}{model} & \textbf{\footnotesize{Remain}} & \textbf{\footnotesize{Dev set}} & \textbf{\footnotesize{DL 19}} & \textbf{\footnotesize{DL 20}} \\
\Xhline{1pt}
\multirow{5}{*}{\makecell{w/o \\ Pruning}} &
BM25 & - & 18.7 & 50.6 & 47.5\\
&SPLADEv2 & - & 35.8 & 70.6 & 68.7 \\ 
&ColBERTv2 \newline \scriptsize{end-to-end} & - & 39.7 & 74.5 & - \\
&ColBERTv2 \newline\scriptsize {reranking} & - & 40.0 & 74.4 & 75.6 \\
\Xhline{1pt}
\multirow{6}{*}{\makecell{w/ \\ Pruning}} &
First $p$ & 50\% & 37.7 & 72.3 & - \\
&IDF & 50\% & 32.6 & 70.2 & - \\
&Stopwords & 67\% & - & 70.9 & 67.8 \\
&Att. Score & 50\% & 36.0 & 72.0 & - \\
&AligneR$_{base}$ & 40\% & 38.1 & - & - \\
&{ConstBERT} & 47\% & 39.0 & 73.1 & 73.3 \\
\Xhline{1pt}
\multicolumn{6}{c}{LP-based pruning with threshold 0.7} \\ \hline 
\multirow{3}{*}{\makecell{Ours \\ $\alpha\mathcal{L}_r^{(L1)}$}}
& $\alpha=0.1$ & 25\% & 37.1 & 72.8$^\dagger$ & 70.6\\
& $\alpha=0.05$ & 46\% & 37.9 & 72.6$^\dagger$ & 71.4\\
& $\alpha=0.01$ & 83\% & 39.2 & 73.6$^\dagger$ & 72.1 \\
\hline
\multirow{5}{*}{\makecell{Ours \\ $\alpha\mathcal{L}_r^{(sim)}$} } 
& $\alpha=2$ & 24\% & 38.4 & 72.5$^\dagger$ & 72.6 \\
& $\alpha=1.25$ & 29\% & 39.3 & 73.7$^\dagger$ & 72.8 \\
& $\alpha=0.8$ & \textbf{32\%} & \textbf{39.7$^\dagger$} & \textbf{73.8$^\dagger$}& \textbf{73.3} \\
& $\alpha=0.5$ & 37\% & 39.5 & 73.9$^\dagger$ & 73.7 \\
& $\alpha=0.1$ & \textbf{52\%} & \textbf{39.9$^\dagger$} & \textbf{74.2$^\dagger$} & \textbf{73.8}\\
\hline
\multirow{3}{*}{\makecell{Our \\ $\alpha\mathcal{L}_r^{(nuc)}$}}
& $\alpha=0.3$ & 38\% & 38.5 & 71.9$^\dagger$ & 73.9$^\dagger$ \\
& $\alpha=0.2$ & 40\% & 39.2 & 73.9$^\dagger$ & 73.1 \\
& $\alpha=0.1$ & 56\% & 39.4 & 73.9$^\dagger$ & 73.2 \\
\Xhline{0.8pt}
\multicolumn{6}{c}{L1-norm based pruning with threshold 0.5} \\
\hline
\multirow{3}{*}{\makecell{Ours \\ $\alpha\mathcal{L}_r^{(L1)}$}}
& $\alpha=0.1$ & 9\% & 37.5 & 73.0$^\dagger$ & 70.6\\
& $\alpha=0.05$ & 14\% & 37.9 & 72.6$^\dagger$ & 71.4\\
& $\alpha=0.01$ & \textbf{31\%} & \textbf{39.2} & \textbf{73.6$^\dagger$} & \textbf{72.1} \\
\hline
\multirow{5}{*}{\makecell{Ours \\ $\alpha\mathcal{L}_r^{(sim)}$} } 
& $\alpha=2$ & 25\% & 38.7 & 72.8$^\dagger$ & 72.6 \\
& $\alpha=1.25$ & 29\% & 39.4 & 73.1$^\dagger$ & 72.9 \\
& $\alpha=0.8$ & 33\% & 39.7$^\dagger$ & 73.8$^\dagger$ & 73.2 \\
& $\alpha=0.5$ & 37\% & 39.5 & 73.9$^\dagger$ & 73.7 \\
& $\alpha=0.1$ & 51\% & 39.9$^\dagger$ & 74.2$^\dagger$ & 73.8\\
\hline
\multirow{3}{*}{\makecell{Our \\ $\alpha\mathcal{L}_r^{(nuc)}$}}
& $\alpha=0.3$ & 18\% & 38.5 & 71.9$^\dagger$ & 73.9$^\dagger$ \\
& $\alpha=0.2$ & 25\% & 39.2 & 73.9$^\dagger$ & 73.1 \\
& $\alpha=0.1$ & 34\% & 39.4 & 73.9$^\dagger$ & 73.2 \\
\Xhline{1.5pt}
\end{tabular}
\end{table}

\begin{table}[h!]
\renewcommand\cellalign{l}
\caption{The LoTTE search results in Success@5 and the remaining ratio of our methods for each sub-dataset. $\dagger$: NO statistical significant performance drop with ($p\leq0.05$) from the result of ColBERTv2 reranking version under the two-tailed Student’s t-test.}  
\begin{tabular}{m{1cm}m{1.6cm}C{0.5cm}C{0.5cm}C{0.5cm}C{0.5cm}C{0.5cm}|C{0.5cm}}
\Xhline{1pt}
\multicolumn{2}{c}{model}  & \textbf{Wrt.} & \textbf{Rcr.} & \textbf{Sci.} & \textbf{Tch.} & \textbf{LS.} & \textbf{Avg.}\\
\Xhline{1pt}
\multirow{5}{*}{\makecell{w/o. \\ Pruning}}  
& BM25  
& 60.3 & 56.5 & 32.7 & 41.8 & 63.8 & 51.0 \\
& SPLADEv2 
& 72.1 & 67.8 & 53.2 & 60.1 & 80.0 & 66.6\\
&ColBERTv2 \newline \scriptsize{end-to-end}
& 80.1 & 72.3 & 56.7 & 66.1 & 84.7 & 72.0\\
&ColBERTv2 \newline \scriptsize{reranking}  
& 80.0 & 73.0 & 56.9 & 66.4 & 84.3 & 72.1\\
\Xhline{1pt}
\multirow{2}{*}{\makecell{\small{Ours} \\ \small{$\alpha\mathcal{L}_r^{(L1)}$}}}
& $\alpha=0.01$ 
& 77.5 & 72.5$^\dagger$ & 57.2$^\dagger$ & 65.4$^\dagger$ & 84.1$^\dagger$ & 71.4\\
& Remain \%  
& \textcolor{gray}{29\%} & \textcolor{gray}{29\%} & \textcolor{gray}{26\%} & \textcolor{gray}{31\%} & \textcolor{gray}{26\%} & \textcolor{gray}{29\%} \\
\Xhline{0.8pt}
\multirow{4}{*}{\makecell{\small{Ours} \\ \small{$\alpha\mathcal{L}_r^{(sim)}$}}}
& $\alpha=0.8$ 
& 77.4 & 72.1$^\dagger$ & 56.9$^\dagger$ & 65.9$^\dagger$ & 84.0$^\dagger$ & 71.3 \\
& Remain \%   
& \textcolor{gray}{34\%} & \textcolor{gray}{35\%} & \textcolor{gray}{35\%} & \textcolor{gray}{37\%} & \textcolor{gray}{34\%} & \textcolor{gray}{35\%} \\
\cline{2-8}{}
& $\alpha=0.1$
& 78.9$^\dagger$  & 72.0$^\dagger$  & 57.2$^\dagger$  & 65.4$^\dagger$  & 84.1$^\dagger$  & 71.3\\
& Remain \%  
& \textcolor{gray}{57\%} & \textcolor{gray}{58\%} & \textcolor{gray}{70\%} & \textcolor{gray}{59\%} & \textcolor{gray}{59\%} & \textcolor{gray}{61\%} \\
\Xhline{1.5pt}
\label{tab:lotte}
\end{tabular}
\end{table}

\begin{table}[h!] 
\renewcommand\cellalign{l}
\caption{The MLM training and the IR oriented training results. $\dagger$: NO statistical significant performance drop with  ($p\leq0.05$) from the result of Distil-ColBERTv2 reranking version under the two-tailed Student’s t-test. } 
\begin{tabular}
{m{1.2cm}m{1.2cm}C{1cm}C{0.8cm}C{0.8cm}}
\Xhline{1pt}
\multicolumn{2}{c}{model} & \textbf{\footnotesize{Remain}} & \textbf{\footnotesize{DL 19}} & \textbf{\footnotesize{DL 20}} \\
\Xhline{1pt}
\multicolumn{2}{l}{Distil-ColBERTv2 \scriptsize {reranking}} 
& - & 73.5 & 73.6 \\
\multicolumn{2}{l}{ColBERTer} 
&  40\% & 72.7 & 73.3 \\ 
\Xhline{1pt}
\multirow{2}{*}{\makecell{ \small{$0.01\mathcal{L}_r^{(L1)}$}}}  
& \small{MLM cp.} 
& 35\% & 73.2$^\dagger$ & 72.6$^\dagger$ \\ 
\cline{3-5}
&\small{IR cp.}
& 32\% & 73.1$^\dagger$ & 72.3$^\dagger$ \\ 
\Xhline{1pt}
\multirow{2}{*}{\makecell{ \small{$0.8\mathcal{L}_r^{(sim)}$}}}  
& \small{MLM cp.} 
& 34\% & 73.3$^\dagger$ & 73.0$^\dagger$ \\
\cline{3-5}
&\small{IR cp.}
& 32\% & 72.4$^\dagger$ & 72.3$^\dagger$ \\
\Xhline{1pt}
\label{tab:scratch_res}
\end{tabular}
\end{table}

\begin{table*}[h!] 
\renewcommand\cellalign{l}
\caption{The BEIR results in nDCG@10 and remaining ratio of our methods for each sub-dataset. $\dagger$: NO statistical significant performance drop with  ($p\leq0.05$) from the result of ColBERTv2 reranking version under the two-tailed Student’s t-test. } \label{tab:beir}
\begin{tabular}
{m{1cm}m{1.6cm}C{1.2cm}C{0.8cm}C{0.8cm}C{0.8cm}C{0.8cm}C{0.8cm}C{0.8cm}C{0.8cm}C{0.8cm}C{0.8cm}C{0.8cm}C{0.8cm}C{0.8cm}C{0.8cm}|C{0.8cm}}
\Xhline{1pt}
\multicolumn{2}{c}{model} & \textbf{Remain} & \textbf{CF} & \textbf{DB} & \textbf{FQ} & \textbf{NF} & \textbf{NQ} & \textbf{QU} & \textbf{SD} & \textbf{SF} & \textbf{TC} & \textbf{TO} & \textbf{Avg.}\\
\Xhline{1pt}
\multirow{5}{*}{\makecell{w/o. \\ Pruning}}  
& BM25 & - 
& 21.3 & 31.3 & 23.6 & 32.5 & 32.9 & 78.9 & 15.8 & 66.5 & 65.6 & 36.7 & 40.5\\
& SPLADEv2 & - 
& 21.3 & 42.0 & 31.6 & 32.8 & 50.8 & 81.1 & 14.0 & 65.9 & 65.5 & 25.5 & 43.0\\
&ColBERTv2 \newline \scriptsize{end-to-end} & - 
& 17.6 & 44.6 & 35.6 & 33.8 & 56.2 & 85.2 & 15.4 & 69.3 & 73.8 & 26.3 & 45.8\\
&ColBERTv2 \newline \scriptsize{reranking} & - 
& 20.3 & 45.9 & 35.7 & 34.3 & 56.8 & 85.7 & 14.4 & 68.5 & 75.5 & 33.7 & 47.0\\
\Xhline{1pt}
\multirow{3}{*}{\makecell{w/. \\ Pruning}}
& First $p$ & 50\%
& 15.5 & 44.0 & 32.4 & 32.4 & 53.2 & 78.7 & 15.2 & 61.6 & 72.1 & 26.1 & 43.1\\
& IDF & 50\%
& 17.0 & 45.1 & 34.0 & 32.6 & 55.0 & 85.3 & 15.5 & 64.8 & 71.4 & 26.2 & 44.7\\
& Att. Score & 50\%
& 16.7 & 44.4 & 33.6 & 32.6 & 54.7 & 84.4 & 15.5 & 64.0 & 70.0 & 26.2 & 44.2\\
\Xhline{1pt}
\multirow{2}{*}{\makecell{\small{Ours} \\ \small{$\alpha\mathcal{L}_r^{(L1)}$}}}
& $\alpha=0.01$ & -
& 18.6 & 44.3 & 34.3 & 33.8$^\dagger$ & 54.0 & 84.5 & 14.2$^\dagger$ & 67.0$^\dagger$ & 71.6 & 31.8 & 45.4\\
& Remain \% & 
& \textcolor{gray}{34\%} & \textcolor{gray}{39\%} & \textcolor{gray}{30\%} & \textcolor{gray}{29\%} & \textcolor{gray}{35\%} & \textcolor{gray}{76\%} & \textcolor{gray}{32\%} & \textcolor{gray}{30\%} & \textcolor{gray}{32\%} & \textcolor{gray}{28\%} & \textcolor{gray}{39\%}\\
\Xhline{0.8pt}
\multirow{4}{*}{\makecell{\small{Ours} \\ \small{$\alpha\mathcal{L}_r^{(sim)}$}}}
& $\alpha=0.8$ & -
& 19.6 & 42.9 & 35.0$^\dagger$ & 33.8$^\dagger$ & 54.7 & 82.6 & 13.9 & 66.9$^\dagger$ & 74.3$^\dagger$ & 33.1$^\dagger$ & 45.7 \\
& Remain \% &  
& \textcolor{gray}{36\%} & \textcolor{gray}{34\%} & \textcolor{gray}{36\%} & \textcolor{gray}{38\%} & \textcolor{gray}{36\%} & \textcolor{gray}{40\%} & \textcolor{gray}{36\%} & \textcolor{gray}{38\%} & \textcolor{gray}{37\%} & \textcolor{gray}{40\%} & \textcolor{gray}{37\%}\\
\cline{2-14}{}
& $\alpha=0.1$ & -
& 20.2$^\dagger$  & 44.7 & 35.0$^\dagger$  & 34.4$^\dagger$ & 55.5 & 84.4 & 14.0 & 67.4$^\dagger$  & 71.1 & 32.8$^\dagger$  & 45.9\\
& Remain \% & 
& \textcolor{gray}{60\%} & \textcolor{gray}{58\%} & \textcolor{gray}{57\%} & \textcolor{gray}{64\%} & \textcolor{gray}{57\%} & \textcolor{gray}{58\%} & \textcolor{gray}{58\%} & \textcolor{gray}{69\%} & \textcolor{gray}{65\%} & \textcolor{gray}{65\%} & \textcolor{gray}{61\%}\\
\Xhline{1.5pt}
\end{tabular}
\end{table*}

\subsection{Out-of Domain Retrieval} \label{sec:OoD}

We now turn to the problem of generalization by evaluating our models on out-of-domain datasets (BEIR and LoTTE). We selected the best configurations, that is, the models having the best trade-off between the pruning ratio and the evaluation score, that is (1) L1 norm regularization ($\alpha=0.01$) with norm threshold of $\theta_{N}=0.5$; (2) document similarity regularization ($\alpha=0.8$) with $\theta_{LP}=0.7$. We also include the checkpoint of document similarity regularization ($\alpha=0.1$) with $\theta_{LP}=0.7$ since this is the best with respect to IR metrics (on MS MARCO dev small).

The results are presented in Tables \ref{tab:beir} and~\ref{tab:lotte}. Note that AligneR~\cite{AligneR} only reports detailed results when not pruning, so we do not include their results here.
In general, our models perform slightly worse than the ColBERTv2 re-ranking version (nDCG@10 45.9 vs. 47.0 on BEIR, and Success@5 71.3 vs. 72.0 on LoTTE even for our best performing model). The difference is larger than the in-domain evaluation tasks (-0.8\% on dev-small vs. -2.8\% on BEIR, and -1.1\% on LoTTE on model trained on document similarity regularization with $\alpha=0.8$). Increasing the remaining ratio around 60\% allows to preserve the performance. 

Furthermore, for the same model, on the out-of-domain task, more tokens remain than the in-domain task~(e.g. 33\% in-domain vs. 39\% on BEIR and 29\% on LoTTE on a model trained on document similarity regularization with $\alpha=0.8$). This was expected since when regularizing we modify the distribution of the tokens: The model mainly learns to distinguish which tokens are important for the in-domain dataset. However, the performance in all cases is still much better than the SPLADEv2 baseline.



\begin{table*}[h!] 
    \centering
    \caption{A Qualitative analysis of our two proposed pruning strategy based on the best model trained on our designed regularizations. The remaining tokens are in \textbf{bold}. }
    \renewcommand{\arraystretch}{1.5}
    \begin{tabular}{m{1.8cm} | m{7cm} | m{7cm}}
        \hline
        \textbf{Example \newline passage} & Norm pruning: $\theta_N=0.5$ & LP pruning: $\theta_{LP} = 0.7$ \\ 
        \hline
        Doc-Sim \newline Regularization &
        What is the \textbf{Population} of \textbf{New Delhi}. According to \textbf{estimated} figures from \textbf{Census} of India, \textbf{Population} of New Delhi in \textbf{2016} is 18.6 million. Delhi's is witnessing a huge \textbf{growth} in its population every \textbf{year}. \textbf{Population} of Delhi \textbf{city} is estimated to cross 25 million in \textbf{2020}. &
        What is the \textbf{Population} of \textbf{New Delhi}. According to \textbf{estimated figures} from \textbf{Census} of India, \textbf{Population} of New Delhi in \textbf{2016} is 18.6 million. Delhi's is witnessing a huge \textbf{growth} in its population every \textbf{year}. Population of Delhi city is estimated to cross 25 million in \textbf{2020}. \\ 
        \hline
        L1-Norm \newline Regularization &
        What is the \textbf{Population} of \textbf{New Delhi}. According to \textbf{estimated figures} from \textbf{Census} of India, \textbf{Population} of New Delhi in \textbf{2016} is 18.6 million. Delhi's is witnessing a huge \textbf{growth} in its population every \textbf{year}. \textbf{Population} of Delhi \textbf{city} is estimated to cross 25 \textbf{million} in \textbf{2020}. &
        What is the \textbf{Population} of \textbf{New Delhi}. According \textbf{to estimated figures} from \textbf{Census of India}, \textbf{Population of New Delhi in 2016} is \textbf{18.6 million}. Delhi's is witnessing a \textbf{huge growth in} its \textbf{population} every \textbf{year}. \textbf{Population} of \textbf{Delhi city} is \textbf{estimated} to \textbf{cross} 25 \textbf{million} in \textbf{2020}. \\
        \hline
    \end{tabular}
    \label{tab:interpretablity}
\end{table*}

\subsection{IR/MLM Checkpoint Training Analysis} \label{sec:Scratch}

In this section, we report experimental results where we compare whether training from an MLM-based checkpoint (i.e., the DistilBERT checkpoint) is better or worse than starting from an IR model checkpoint (i.e., a DistilBERT-based ColBERTv2 checkpoint). 
We use DistilBERT~\cite{distilbert}  \href{https://huggingface.co/distilbert/distilbert-base-uncased}{checkpoint (HuggingFace)} since it is much smaller than a BERT model (2x speed gain) while providing competitive results~\cite{SPLADEv2, tas-balanced}. 

We fine-tune DistilBERT using the ColBERTv2 IR loss only (Eq.~\eqref{eq:ir-loss}), resulting in a model we name DistilColBERT. From this checkpoint, we fine-tune our model using our proposed regularizations (IR checkpoint), and compare this to a model trained from the start with the different regularizations (MLM checkpoint).
%
Following the previous sections, we select the model with the best harmonic mean (H2) between the pruning ratio and MRR@10. 

The results are reported in Table~\ref{tab:scratch_res} and show that there is no clear best methodology. For the two models tested, the pruning ratio is higher when using the IR checkpoint, but the overall performance (DL19 and DL20) is worse. There is a slight advantage in using an MLM checkpoint with similarity-based regularization, but more experiments are needed to validate this point.


\subsection{Interpretability Analysis} \label{sec:interpretablity}

To conclude our section of experimental results, in Table \ref{tab:interpretablity}, we show an example of different pruning for a passage from MS MARCO v1. 
In general, we observe that the remaining tokens are mostly the important nouns and verbs, which is consistent with human intuition that these tokens play a more important role in the semantic matching.
We also note that, for similarity-based regularization, LP-based pruning removes the second occurrence of ``Population'', but norm-based pruning does not: this showcases the benefit of using LP-based pruning. Likewise, ``city'' is discarded with LP-based pruning as its semantic meaning is already covered by ``New Dehli''.


\section{Conclusion, Limitations and Future Work} \label{sec:conclusion}

In this paper, we proposed a principled framework, that of vector dominance, that defines exactly which tokens can be removed from the set of document vectors in a late interaction model like ColBERT, without any impact on the score function. Based on this concept, we adapted the ColBERT model scoring function and proposed various regularization and pruning strategies. Through our experiments, we found that our methods can remove around 70\% of document vectors while preserving good both in-domain and out-of-domain performance. However, our work still has some limitations that we discuss in the following. 

First, our current pipeline relies on the SPLADEv2 reranking. Although it provides a good result, the reranking pipeline requires more index storage than the end-to-end scenario. Adapting recent works like~\cite{EMVB} to our use case (i.e. handling varying norm document tokens) would allow the model to reach a better space efficiency.


Second, although we have provided some techniques to lower to computational cost of the LP procedure, they are still too costly. However, we did not leverage the fact that our LP problem has a lot of symmetries, since for the same document the problems are very similar. An alternative would be to design approximate pruning where the approximation is controlled -- and this could be leveraged during training and pruning.

Third, regularization strategies are loosely related to the ratio of dominance, which should be the main target. Designing better regularization functions could substantially increase the performance of our methodology. This could be combined with adapting the Transformer architecture so it is more prone to produce dominated vectors, or by using better pretraining than MLM -- better pre-training can lead to substantial performance increase~\cite{gaoCondenserPretrainingArchitecture2021,maEnhancingSparseRetrieval2023}.
 




\begin{acks}
The authors acknowledge ANR - FRANCE (French National Research Agency) for its financial support of the GUIDANCE project n°ANR-23-IAS1-0003. 
This work was granted access to the HPC resources of IDRIS under the allocation made by GENCI.
\end{acks}
\bibliographystyle{ACM-Reference-Format}
\bibliography{main.bib}

\end{document}